\documentclass[11pt]{article}
\usepackage{calc}
\usepackage[centertags]{amsmath}
\usepackage{latexsym}
\usepackage{amsfonts}
\usepackage{amssymb}
\usepackage{graphicx}
\usepackage{amsthm}
\usepackage{fancyhdr}
\usepackage [dvips]{epsfig}

\usepackage{newlfont}
\usepackage[latin1]{inputenc}
\usepackage[french,english]{babel}
\usepackage{graphicx}
\usepackage{t1enc}
\usepackage[matrix,arrow,frame,curve]{xy}
\theoremstyle{plain}

\newtheorem{prop}{Proposition}[section]
\newtheorem{theo}{Theorem}[section]
\newtheorem{definition}{Definition}[section]
\newtheorem{cor}{Corollary}[section]
\newtheorem{rem}{Remark}[section]

\theoremstyle{plain}

\title{Integer solutions of integral inequalities and $H$-invariant Jacobian Poisson structures}

\author{ G. Ortenzi \footnote{
 Dipartimento di Matematica Pura e Applicazioni,
  Universit\`a degli Milano Bicocca
  Via R.Cozzi, 53
 20125, Milano, Italia
  E-mail address: giovanni.ortenzi@unimib.it},
    V. Rubtsov \footnote{
  Laboratoire Angevin de Recherche en Math\'ematiques
  Universit\'e D'Angers, D\'epartement de Math\'ematiques
  2, boulevard Lavoisier, 49045 Angers, France
  E-mail address: Volodya.Roubtsov@univ-angers.fr },
  S.R. Tagne Pelap\footnote{
  Mathematics Research Unit at Luxembourg, University of Luxembourg,
  6 rue Richard Coudenhove-Kalergi, L-1359 Luxembourg City, Grand-Duchy of Luxembourg,
  E-mail address: serge.pelap@uni.lu}}

\begin{document}

\maketitle

\abstract{We study the Jacobian Poisson structures in any dimension invariant with respect to the discrete Heisenberg group. The classification problem is related to the discrete volume of suitable solids. Particular attention is given to dimension $3$ whose simplest example is the Artin-Schelter-Tate Poisson tensors.}
\section*{Introduction}
This paper continues the author's  program of studies the Heisenberg invariance properties of polynomial Poisson algebras which were started in \cite{odru} and extended in \cite{ORT1, ORT2}. Formally speaking, we consider the polynomials in $n$ variables $\mathbb C[x_0, x_1,\cdots, x_{n-1}]$ over $\mathbb C$ and the action of some subgroup $H_n$ of $GL_n(\mathbb C)$ generated by the shifts operators $x_i\longrightarrow x_{i+1}$ $(mod {\mathbb Z}_n)$ and by the operators $x\longrightarrow \varepsilon^ix_i,$ where $\varepsilon^n=1.$ We are interested in the polynomial Poisson brackets on $\mathbb C[x_0, x_1,\cdots, x_{n-1}]$ which are "stable" under this actions ( we will give more precise definition below).\\
The most famous examples of the Heisenberg invariant polynomial Poisson structures are the Sklyanin-Odesskii-Feigin -Artin-Tate quadratic Poisson brackets known also as the elliptic Poisson structures. One can also think about these algebras like the"quasi-classical limits" of elliptic Sklyanin associative algebras. These is a class of Noetherian graded associative algebras which are Koszul, Cohen-Macaulay and have the same Hilbert function
as a polynomial ring with $n$ variables. The above mentioned Heisenberg group action provides the automorphisms
of Sklyanin algebras which are compatible with the grading and defines an $H_n$-action
on the elliptic quadratic Poisson structures on ${\mathbb P}^n$. The latter are identified with Poisson structures on some moduli
spaces of the degree $n$ and rank $k + 1$ vector bundles with parabolic structure (=
the flag $0\subset F\subset {\mathbb C}^{k+1}$ on the elliptic curve $\cal E$). We will denote this elliptic Poisson
algebras by $q_{n;k}(\cal E)$. The algebras $q_{n;k}(E)$ arise in the Feigin-Odesskii "deformational"
approach and form a subclass of polynomial Poisson structures. A comprehensive review of elliptic algebras can be found in (\cite{ode3}) to which we
refer for all additional information. We will mention only that as we have proved in \cite{ORT2} all elliptic Poisson algebras (being in particular Heisenberg-invariant) are unimodular.

Another interesting class of polynomial Poisson structures consists of so-called Jacobian Poisson structures (JPS). These structures are a special case of Nambu-Poisson structures. Their rank is two and the Jacobian Poisson bracket $\left\{P,Q\right\}$ of two polynomials $P$ and $Q$ is given by the determinant of Jacobi matrix of functions $(P, Q, P_1, . . . , P_{n-2}).$ The polynomials $P_i, 1\leq i \leq n-2$ are Casimirs of the bracket and
under some mild condition of independence are generators of the centrum for the Jacobian Poisson algebra structure on $\mathbb C[x_0,\ldots,x_{n-1}]$. This type of Poisson algebras was intensively studied (due to their natural origin and relative simplicity) in a huge number of publications among which we should mention \cite{tak},\cite{nam}\cite{khi1},\cite{khi2}, \cite{pelap1} and \cite{odru}.

There are some beautiful intersections between two described types of polynomial Poisson structures: when we are restricting ourselves to the class of quadratic Poisson brackets then there are only Artin-Schelter-Tate ($n=3$) and Sklyanin ($n=4$) algebras which are both elliptic and Jacobian. It is no longer true for $n>4$. The relations between the Sklyanin Poisson algebras $q_{n,k}(\cal E)$ whose centrum has dimension 1 (for $n$ odd) and 2 (for $n$ even) in the case $k=1$ and is generated by $l = {\mathfrak gcd}(n, k + 1)$ Casimirs for $q_{n,k}(\cal E)$ for $k>1$ are in general quite obscure. We can easily found that sometimes the JPS structures correspond to some degenerations of the Sklyanin elliptic algebras. One example of such JPS
for $n=5$ was remarked in \cite{khi2} and was attributed to so-called Briesckorn-Pham polynomials for $n=5$:
$$P_1 = \sum_{i=0}^4\alpha_i x_i;\,  P_2 = \sum_{i=0}^4\beta_i x_i^2;\, P_3 = \sum_{i=0}^4\gamma_i x_i^2.$$
It is easy to check that the homogeneous quintic $P=P_1P_2P_3$ (see sect. 3.2) defines a Casimir for  some "rational degeneration of (one of)
elliptic algebras $q_{5,1}(\cal E)$ and $q_{5,2}(\cal E)$ if it satisfies the $H-$invariance condition.

In this paper we will study the Jacobian Poisson structures in any number of variables which are Heisenberg-invariant and we relate all such structures to some graded sub-vector space $\mathcal H$ of polynomial algebra. This vector space is completely determined by some enumerative problem of a number-theoretic type. More precisely, the
homogeneous subspace $\mathcal H_i$ of $\mathcal H$ of degree $i$ is in bijection with integer solutions of a system of Diophant inequalities. Geometric interpretation of the dimension of $\mathcal H_i$ is described in terms of integer points in a convex polytope given by this Diophant system. In the special case of dimension 3, $\mathcal H$ is a subalgebra of polynomial algebra with 3 variables and all JPS are given by this space. We solve explicitly the enumerative problem in this case  and obtain a complete classification of the H-invariant not necessarily quadratic Jacobian Poisson algebras with three generators. As a by product we explicitly compute the Poincar\'e series of $\mathcal H.$ In this dimension we observe that the H-invariant JPS of degree 5 is given by the Casimir sextic
\begin{equation*}
\begin{split}
P^{\vee}=&\frac{1}{6}a(x_0^6+x_1^6+x_2^6)+\frac{1}{3}b(x_0^3x_1^3+x_0^3x_2^3+ x_1^3x_2^3)+c(x_0^4x_1x_2+x_0x_1^4x_2+x_0x_1x_2^4)\\ &+\frac{1}{2}dx_0^2x_1^2x_2^2
\end{split}
\end{equation*}
$a, b, c, d\in\mathbb C.$ This structure is a "projectively dual" to the Artin-Schelter-Tate elliptic Poisson structure which is the $H$-invariant JPS given by the cubic
\begin{equation*}
P=(x_0^3+x_1^3+x_2^3)+\gamma x_0x_1x_2
\end{equation*}
where $\gamma\in\mathbb C.$
In fact the algebraic variety ${\mathcal E}^{\vee}: \{P^{\vee}=0\}\in \mathbb P^2$ is the (generically) projectively dual to the elliptic curve ${\mathcal E}: \{P=0\}\subset \mathbb P^2$.

The paper is organized as follows: in Section 1 we remind a definition of the Heisenberg group in the Schroedinger representation and describe its action on Poisson polynomial tensors and also the definition of JPS. In Section 2 we treat the above mentioned enumerative problem in dimension 3. The last section concerns the case of any dimension. Here we discuss some possible approaches to the general enumerative question.

\section{Preliminary facts}
Throughout of this paper, $K$ is a field of characteristic zero. Let us start by remaining some elementary notions of the Poisson geometry.
\subsection{Poisson algebras and Poisson manifold}
 Let $\mathcal R$ be a commutative $K$-algebra. One says that $\mathcal R$ is a Poisson algebra if $\mathcal R$ is endowed with a Lie bracket, indicated with $\{ \cdot , \cdot \}$, which is also a biderivation.
One can also say that $\mathcal R$ is endowed with a Poisson structure and therefore the bracket $\{\cdot, \cdot\}$ is called the Poisson bracket.
Elements of the center are called Casimirs: $a\in\mathcal R$ is a Casimir if $\{a,b\}=0$ for all $b\in\mathcal R.$\\
A Poisson manifold $M$ (smooth, algebraic,...) is  a manifold whose  function algebra $\mathcal A$ ($C^\infty(M)$, regular etc.) is endowed with a Poisson bracket.\\
As examples of Poisson structures let us consider a particular subclass of Poisson structures which are uniquely characterized by their Casimirs.  In the dimension $4$ let
$$q_1=\frac{1}{2}(x_0^2 + x_2^2)+kx_1x_3,$$
$$q_2=\frac{1}{2}(x_1^2 + x_3^2)+kx_0x_2,$$
be two elements of $\mathbb C[x_0, x_1, x_2, x_3]$ where $k\in\mathbb C.$\\
On $\mathbb C[x_0, x_1, x_2, x_3]$ a Poisson structure $\pi$ is defined by the formula:
$$\{f, g\}_{\pi}:=\frac{df\wedge dg\wedge dq_1\wedge dq_2}{dx_0\wedge dx_1\wedge dx_2\wedge dx_3}$$
or more explicitly (mod $\mathbb Z_4$):
\begin{eqnarray*}
\{x_i, x_{i+1}\}&=&k^2x_ix_{i+1}-x_{i+2}x_{i+3}, \\
\{x_i, x_{i+2}\}&=&k(x_{i+3}^2-x_{i+1}^2),\ i=0, 1, 2, 3.
\end{eqnarray*}
Sklyanin had introduced this Poisson algebra which carried today his name in a Hamiltonian approach to the continuous and discrete integrable Landau-Lifshitz models (~\cite{sky1}, ~\cite{sky2}). He showed that the Hamiltonian structure of the classical model is completely determined by two  quadratic "Casimirs". The Sklyanin Poisson algebra is also called elliptic since its relations with an elliptic curve. The elliptic curve enters in the game from the  geometric side. The symplectic foliation of Sklyanin's structure is too complicated. This is because the structure is degenerated and looks quite different from a symplectic one. But the intersection locus of two Casimirs in the affine space of dimension four (one can consider also the projective situation) is an elliptic curve $\mathcal E$ given by two quadrics $q_{1,2}.$ We can think about this curve $\mathcal E$ as a complete intersection of the couple $q_1=0$, $q_2=0$ embedded in $\mathcal CP^3$ (as it was observed in Sklyanin's initial paper).\\
A possible generalization one can be obtained considering $n-2$ polynomials $Q_i$ in $K^n$ with coordinates $x_i$, $i=0,...,n-1.$ We can define a bilinear differential operation:
$$\{,\} : K[x_1,...,x_n]\otimes K[x_1,...,x_n]\longrightarrow K[x_1,...,x_n]$$
by the formula
\begin{equation}\label{jps}
\{f,g\}=\frac{df\wedge dg\wedge dQ_1\wedge...\wedge dQ_{n-2}}{dx_1\wedge dx_2\wedge...\wedge dx_n},\space\space f,g\in K[x_1,...,x_n]
 \end{equation}
This operation, which gives a Poisson algebra structure on $K[x_1,...,x_n]$, is called a Jacobian Poisson structure (JPS), and it is a partial case of more general $n-m$-ary Nambu operation given by an  antisymmetric $n-m$-polyvector field introduced by Y.Nambu \cite{nam} and was extensively studied by L. Takhtajan \cite{tak}.\\
The polynomials $Q_i, i=1,...,n-2$ are Casimir functions for the brackets (\ref{jps}).\\ \mbox{} \\
There exists a second generalization of the Sklyanin algebra that we will describe briefly in the next sub-section (see for details \cite{ode3}).
\subsection{Elliptic Poisson algebras $q_{n}(\mathcal E, \eta)$ and $q_{n, k}(\mathcal E, \eta)$}
(We report here this subsection from \cite{ORT1} for sake of self-consistency).\\
 These algebras, defined by Odesskii and Feigin, arise as quasi-classical limits of elliptic associative algebras $Q_n(\mathcal E, \eta)$ and $Q_{n, k}(\mathcal E, \eta)$ \cite{odfe1, ode1}.\\
Let $\Gamma=\mathbb Z+\tau\mathbb Z\subset\mathbb C,$ be an integral lattice generated by $1$ and $\tau\in\mathbb C,$ with $\textnormal{Im}\tau > 0.$
Consider the elliptic curve $\mathcal E=\mathbb C/\Gamma$ and a point $\eta$ on this curve.\\
In their article \cite{ode1}, given $k< n,$ mutually prime, Odesskii and Feigin construct an algebra, called  elliptic, $Q_{n, k}(\mathcal E, \eta),$ as an algebra defined by $n$ generators  $\{x_i, i\in\mathbb Z/n\mathbb Z\}$ and the following relations
 \begin{equation}
  \displaystyle\sum_{r\in\mathbb{Z}/n\mathbb{Z}}\frac{\theta_{j-i+r(k-1)}(0)}
 {\theta_{kr}(\eta)\theta_{j-i-r}(-\eta)}x_{j-r}x_{i+r}=0, \ \ i\neq j, i, j\in\mathbb Z/n\mathbb Z
 \end{equation}
 where $\theta_{\alpha}$ are theta functions \cite{ode1}.
%We denote $Q_n(\mathcal E, \eta)$ the algebra $Q_{n, 1}(\mathcal E, \eta).$
\\
 These family of algebras has the following properties :
 \begin{enumerate}
   \item The center of the algebra $Q_{n, k}(\mathcal E, \eta),$ for generic $\mathcal E$ and $\eta,$ is the algebra of polynomial of $m=pgcd(n, k+1)$ variables of degree $n/m;$
   \item $Q_{n, k}(\mathcal E, 0)=\mathbb C[x_1, \cdots, x_n]$ is commutative;
   \item $Q_{n, n-1}(\mathcal E, \eta)=\mathbb C[x_1, \cdots, x_n]$ is commutative for all $\eta$;
   \item $Q_{n, k}(\mathcal E, \eta)\simeq Q_{n, k'}(\mathcal E, \eta),$ if $kk'\equiv 1$ (mod $n$);
   \item the maps $x_i\mapsto x_{i+1}$ et $x_i\mapsto \varepsilon^ix_i$, where $\varepsilon^n=1,$ define  automorphisms of the algebra $Q_{n, k}(\mathcal E, \eta)$;
   \item the algebras $Q_{n,k}(\mathcal E, \eta)$ are deformations of polynomial algebras. The associated Poisson structure is denoted by $q_{n,k}(\mathcal E, \eta)$;
   \item among the algebras $q_{n,k}(\mathcal E, \eta),$ only $q_{3}(\mathcal E, \eta)$ (the Artin-Schelter-Tate algebra) and the Sklyanin algebra $q_{4}(\mathcal E, \eta)$ are Jacobian Poisson structures.
 \end{enumerate}
\subsection{The Heisenberg invariant Poisson structures}
\subsubsection{The G-invariant Poisson structures}
Let $G$ be a group acting on a Poisson algebra $\mathcal R.$
\begin{definition}
A Poisson bracket $\{\cdot, \cdot\}$ on $\mathcal R$ is said to be a $G$-invariant if $G$ acts on $\mathcal R$ by Poisson automorphisms.\\
In other words, for every $g\in G$ the morphism $\varphi_g : \mathcal R\longrightarrow\mathcal R, a\mapsto g\cdot a$ is an automorphism and the following diagram is a commutative:
$$\xymatrix{\mathcal R\times\mathcal R \displaystyle{\ar[r]^{\varphi_g\times \varphi_g}}\displaystyle{\ar[d]_{\{\cdot, \cdot\}}}&
\mathcal R\times\mathcal R\displaystyle{\ar[d]^{\{\cdot, \cdot\}}}\\
\mathcal R \displaystyle{\ar[r]^{\varphi_g}}&\mathcal R}.$$
\end{definition}
\subsubsection{The H-invariant Poisson structures}
In their paper \cite{ORT1}, the authors introduced the notion of $H$-invariant Poisson structures.
That is a special case of a $G$-invariant structure when $G$ in the finite Heisenberg Group and $\mathcal R$ is the polynomial algebra.
  Let us remind this notion.\\
Let $V$ be a complex vector space of dimension $n$ and $e_0,\cdots, e_{n-1}$ - a basis of $V$. Take the $n-$th primitive root of unity $\varepsilon=e^{\frac{2\pi i}{n}}.$\\
Consider $\sigma, \tau$ of $GL(V)$ defined by:
$$\sigma(e_m)=e_{m-1};$$
$$\tau(e_m)=\varepsilon^me_m.$$
The  Heisenberg of dimension $n$ is nothing else that the subspace $H_n\subset GL(V)$ generated by $\sigma$ and $\tau.$\\
From now on, we assume $V=\mathbb C^n,$ with $x_0, x_1,\cdots, x_{n-1}$ as a basis and consider the coordinate ring $\mathbb C[x_0,x_1,\cdots, x_{n-1}].$\\
Naturally, $\sigma$ and $\tau$ act by automorphisms on the algebra $\mathbb C[x_0,x_1,\cdots, x_{n-1}]$ as follows:
\begin{eqnarray*}
\sigma\cdot(\alpha x_0^{\alpha_0}x_1^{\alpha_1}\cdots x_{n-1}^{\alpha_{n-1}})=\alpha x_0^{\alpha_{n-1}}x_1^{\alpha_0}\cdots x_{n-1}^{\alpha_{n-2}};\\
\tau\cdot(\alpha x_0^{\alpha_0}x_1^{\alpha_1}\cdots x_{n-1}^{\alpha_{n-1}})=\varepsilon^{\alpha_1+2\alpha_2+\cdots+(n-1)\alpha_{n-1}}\alpha x_0^{\alpha_0}x_1^{\alpha_1}\cdots x_{n-1}^{\alpha_{n-1}}.
\end{eqnarray*}
We introduced in \cite{ORT1} the notion of $\tau$-degree on the polynomial algebra  $\mathbb C$[$x_0$,$x_1$,$\cdots$, $x_{n-1}$]. The $\tau$-degree of a monomial $M=\alpha x_0^{\alpha_0}x_1^{\alpha_1}\cdots x_{n-1}^{\alpha_{n-1}}$ is the positive integer $\alpha_1+2\alpha_2+\cdots+(n-1)\alpha_{n-1}\in\mathbb Z/n\mathbb Z$  if $\alpha\neq0$ and $-\infty$ if not. The $\tau$-degree of $M$ is denoted $\tau \varpi(M)$. A $\tau$-degree of a polynomial is the highest $\tau$-degree of its monomials.\\ \mbox{} \\
For simplicity the $H_n$-invariance condition will be referred from now on just as $H$-invariance. An $H$-invariant Poisson bracket on $\mathcal A=\mathbb C[x_0,x_1,\cdots, x_{n-1}]$ is nothing but a bracket on $\mathcal A$ which satisfy the following:
$$\{x_{i+1}, x_{j+1}\}=\sigma\cdot\{x_i, x_j\}$$
 $$\tau\cdot\{x_i, x_j\}=\varepsilon^{i+j}\{x_i, x_j\},$$
for all $i, j\in\mathbb Z/n\mathbb Z.$\\ \mbox{} \\
The $\tau$ invariance is, in some sense, a ''discrete'' homogeneity.
\begin{prop}\cite{ORT1}
The Sklyanin-Odesskii-Feigin Poisson algebras $q_{n,k}(\mathcal E)$ are $H$-invariant Poisson algebras.
\end{prop}
Therefore an $H$-invariant Poisson structures on the polynomial algebra $\mathcal R$ includes as the Sklyanin Poisson algebra or more generally of the Odesskii-Feigin Poisson algebras.\\
In this paper, we are interested by the intersection of the two classes of generalizations of Artin-Shelter-Tate- Sklyanin Poisson algebras: JPS and $H$-invariant Poisson structures.
\begin{prop}\cite{ORT1}\label{degcond}
If $\{\cdot, \cdot\}$ is an $H$-invariant polynomial Poisson bracket, the usual polynomial degree of the monomial of $\{x_i, x_j\}$ equals to $2+sn,\ s \in\mathbb N.$
\end{prop}
\begin{prop}\cite{ORT1}\label{Hder}
Let $P\in \mathcal R =\mathbb C[x_0,\cdots, x_{n-1}].$ \\
For all $i\in\{0, \cdots, n-1\},$ $$\sigma\cdot\frac{\partial F}{\partial x_i}= \frac{\partial(\sigma\cdot F)}{\partial(\sigma\cdot x_i)}.$$
\end{prop}
\section{$H$-invariant JPS in dimension 3}
We consider first a generalization of Artin-Schelter-Tate quadratic Poisson algebras.
Let $\mathcal R=\mathbb C[x_0, x_1, x_2]$ be the polynomial algebra with 3 generators. For every $P\in\mathcal A$
we have a JPS $\pi(P)$ on $\mathcal R$ given by the formula:
$$\{x_i, x_j\}=\frac{\partial P}{\partial x_k}$$
where $(i, j, k)\in{\mathbb Z}/{3\mathbb Z}$ is a cyclic permutation of $(0, 1, 2).$ Let $\mathcal H$ the set of all $P\in\mathcal A$
such that $\pi(P)$ is an $H$-invariant Poisson structure.
\begin{prop}
If $P\in\mathcal H$ is a homogeneous polynomial, then $\sigma\cdot P=P$ and $\tau\varpi(P)=0.$
\end{prop}
\begin{proof}
Let $(i, j, k)\in(\mathbb Z/3\mathbb Z)^3$ be a cyclic permutation of $(0, 1, 2).$ One has:
\begin{equation}\label{brack3}
\{x_i, x_j\}=\frac{\partial P}{\partial x_k}
\end{equation}

and
\begin{equation}
\{x_{i+1}, x_{j+1}\}=\frac{\partial P}{\partial x_{k+1}}.
\end{equation}
Using the proposition (\ref{Hder}), we conclude that for all $m\in\mathbb Z/3\mathbb Z,$ $\frac{\partial\sigma P}{\partial x_m}=\frac{\partial P}{\partial x_m}.$\\
It gives that $\sigma\cdot P=P.$\\
On the other hand, from  the equation (\ref{brack3}), one has $\tau$-$\varpi(P)\equiv i+j+k \ mod \ 3.$ And we get the second half of the proposition.
\end{proof}
\begin{prop}
$\mathcal H$ is a subalgebra of $\mathcal R.$
\end{prop}
\begin{proof}
Let $P, Q\in\mathcal H.$ It is clear that for all $\alpha, \beta\in\mathbb C,$ $\alpha P+\beta Q$ belongs to $\mathcal H.$\\
Let us denote by $\{\cdot, \cdot\}_F$ the JPS associated to the polynomial $F\in\mathcal R.$ It is easy to verify that $\{x_i, x_j\}_{PQ}=P\{x_i, x_j\}_{Q}+Q\{x_i, x_j\}_{P}.$ Therefore, it is clear that the $H$-invariance condition is verified for the JPS associated to the polynomial $PQ.$
\end{proof}
We endow $\mathcal H$ with the usual grading of the polynomial algebra $\mathcal R.$
For $F$ an element of $\mathcal R,$ we denote by $\varpi(F)$ its usual weight degree.
We denote by $\mathcal H_i$ the homogeneous subspace of $\mathcal H$ of degree $i.$
\begin{prop}\label{cond3}
If 3 does not divide $i\in\mathbb N$ (in other words $i\neq 3k$) then $\mathcal H_i=0.$
\end{prop}
\begin{proof}
First of all $\mathcal H_0=\mathbb C.$ We suppose now $i\neq 0.$ Let $P\in\mathcal H_i,$ $P\neq 0.$ Then $\varpi(P)=i.$ It follows from Proposition (\ref{degcond}) and the definition of the Poisson brackets, that there exists $s\in\mathbb N$ such that $\varpi({x_i, x_j})=2+3s.$ The result follows from the fact that $\varpi({x_i, x_j})=\varpi(P)-1.$
\end{proof}
Set $P=\sum c_{\bar \alpha} x_0^{\alpha_0}x_1^{\alpha_1}x_2^{\alpha_2}$ where $\bar\alpha=(\alpha_0,\alpha_1,\alpha_2).$ We suppose that $\varpi(P)=3(1+s).$
We want to find all $\alpha_0, \alpha_1, \alpha_2$ such that $P\in\mathcal H$ and therefore the dimension $\mathcal H_{3(1+s)}$
as $\mathbb C$-vector space.
\begin{prop}
There exist $s'$, $s''$ and $s'''$ such that
\begin{equation}\label{hjps3}
\left\{\begin{array}{ccl}
\alpha_0+\alpha_1+\alpha_2&=&3(1+s)\\
0\alpha_0+\alpha_1+2\alpha_2&=&3s'\\
\alpha_0+2\alpha_1+0\alpha_2&=&3s''\\
2\alpha_0+0\alpha_1+1\alpha_2&=&3s'''
\end{array}\right.
\end{equation}
\end{prop}
\begin{proof}
This is a direct consequence of proposition (\ref{cond3}).
\end{proof}
\begin{prop}
The system equation (\ref{hjps3}) has as solutions the following set:
\begin{equation}
\left\{\begin{array}{ccr}
\alpha_0&=&4r-2s'-s''\\
\alpha_1&=&-2r+s'+2s''\\
\alpha_2&=&r+s'-s''
\end{array}\right.
\end{equation}
where $r=1+s,$ $s'$ and $s''$ live in the polygon given by the following inequalities in $\mathbb R^2:$
\begin{equation}\label{poly3}
\left\{\begin{array}{ccr}
x+y&\leqslant&3r\\
2x+y&\leqslant&4r\\
-x-2y&\leqslant&-2r\\
-x+y&\leqslant&r\\
\end{array}\right.
\end{equation}
\end{prop}
\begin{rem}
For $r=1,$ one obtains the Artin-Schelter-Tate Poisson algebra which is the JPS given by the Casimir $P=\alpha(x_0^3+x_1^3+x_2^3)+\beta x_0x_1x_2,$ $\alpha, \beta\in\mathbb C.$ Let suppose $\alpha\neq0,$ then it can take the form:
\begin{equation}\label{artcubic}
P=(x_0^3+x_1^3+x_2^3)+\gamma x_0x_1x_2
\end{equation}
where $\gamma\in\mathbb C.$ The interesting feature of this Poisson algebra is that their polynomial
character is preserved even after the following non-algebraic changes of variables:
Let
\begin{equation}\label{npc}
y_0=x_0; \ y_1=x_1x_2^{-1/2}; \ y_2=x_2^{3/2}.
\end{equation}
The polynomial P in the coordinates $(y_0; y_1; y_2)$ has the form
\begin{equation}\label{art6}
\widetilde{P}=(y_0^3+y_1^3y_2+y_2^2)+\gamma y_0y_1y_2.
\end{equation}
The Poisson bracket is also polynomial (which is not evident at all!)
and has the same form:
\begin{equation}
\{x_i, x_j\}=\frac{\partial \widetilde{P}}{\partial x_k}
\end{equation}
where  $(i, j, k)$ is the cyclic permutation of $(0, 1, 2).$
This JPS structure  is no longer satisfied the Heisenberg invariance condition.
But it is invariant with respect the following toric action:
$({\mathbb C}^*)^3 \times {\mathbb P^2}\to {\mathbb P}$ given by the formula:
$$\lambda\cdot(x_0:x_1:x_2)=(\lambda^2 x_0:\lambda x_1:\lambda^3 x_2)$$

Put $deg\ y_0=2; \ deg\ y_2 = 1; \ deg\ y_2 = 3.$ Then the polynomial $\widetilde{P}$ is
also homogeneous in $(y_0; y_1; y_2)$ and defines an elliptic curve
$\widetilde{P}= 0$ in the weighted projective space ${\mathbb WP}_{2;1;3}.$

The similar change of variables
\begin{equation}\label{npc1}
z_0 = x_0^{-3/4}x_1^{3/2};\ z_1 = x_0 ^{1/4} x_1 ^{-1/2} x_2;\ z_2 = x_0 ^{3/2}
\end{equation}
defines the JPS structure invariant with respect to the torus action
$({\mathbb C}^*)^3 \times {\mathbb P^2}\to {\mathbb P}$ given by the formula:
$$\lambda\cdot(x_0:x_1:x_2)=(\lambda x_0:\lambda x_1:\lambda^2 x_2)$$
and related to the elliptic curve $1/3 (z_2 ^2  + z_0 ^2 z_2 + z_0 z_1 ^3 )+kz_0 z_1 z_2 = 0$ in the weighted projective space
 ${\mathbb WP}_{1;1;2}.$

 These structures had appeared in \cite{odru}, their Poisson cohomology were studied by A. Pichereau (\cite{pic}) and their relation to the
 non-commutative del Pezzo surfaces and Calabi-Yau algebras were discussed in \cite{ginzet}.
\end{rem}
\begin{prop}
The subset of $\mathbb R^2$ given by the system (\ref{poly3}) is a triangle $\mathcal T_r$ with $(0, r),$ $(r, 2r)$ and $(2r, 0)$ as vertices.
Then $dim\mathcal H_{3r}=card(\mathcal T_r\cap\mathbb N^2).$
\end{prop}
\begin{figure}[h!]
\centering
\includegraphics[width=8cm, height=8cm]{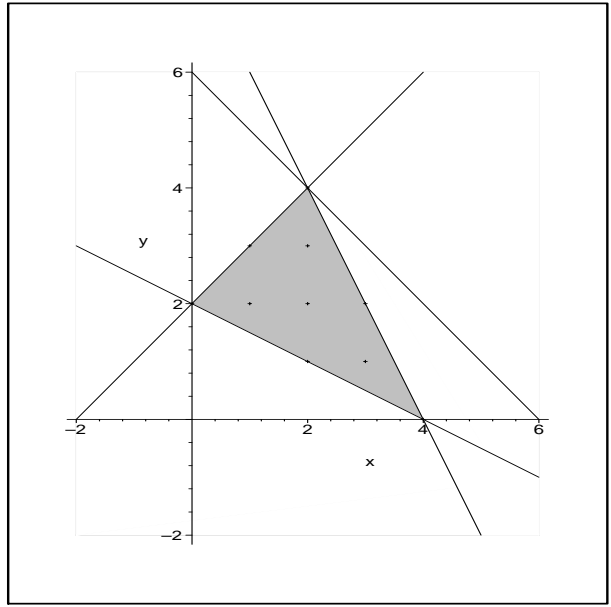}
\caption{An example of the triangle $\mathcal{T}_2$ in the case $r=2$}\label{fig:Triangle}
\end{figure}
\newpage
\begin{rem}
For $r=2,$ the case of above figure, the generic Heisenberg-invariant JPS is given by the sextic polynomial
\begin{equation*}
P^{\vee}=\frac{1}{6}a(p_0^6+p_1^6+p_2^6)+\frac{1}{3}b(p_0^3p_1^3+p_0^3p_2^3+ p_1^3p_2^3)+c(p_0^4p_1p_2+p_0p_1^4p_2+p_0p_1p_2^4)+\frac{1}{2}dp_0^2p_1^2p_2^2
\end{equation*}
$a, b, c, d\in\mathbb C.$ \\
The corresponding Poisson bracket takes the form:
\begin{equation*}
\{p_i, p_j\}=ap_k^5+b(p_i^3+p_j^3)p_k^2+cp_ip_j(p_i^3+p_j^3+4p_k^3)+dp_i^2p_j^2p_k^2,
\end{equation*}
where $i, j, k$ are the cyclic permutations of $0, 1, 2.$\\
This new JPS should be considered as the "projectively dual" to the Artin-Schelter-Tate JPS since the algebraic variety ${\mathcal E}^{\vee}: P^{\vee}=0$ is generically the projective dual curve in $\mathbb P^2$ to the elliptic curve \\
${\mathcal E}: P=(x_0^3+x_1^3+x_2^3)+\gamma x_0x_1x_2=0$.\\
To establish the exact duality and the explicit values of the coefficients we should use (see \cite{GKZ} ch.1) the Schl\"{a}fli's formula for the dual of a smooth
plane cubic ${\mathcal E} =0\subset \mathbb P^2$. The coordinates $(p_0:p_1:p_2)\in {\mathbb P^2}^*$ of a point $p\in {\mathbb P^2}^*$ satisfies
to the sextic relation ${\mathcal E}^{\vee}=0$ iff the line $x_0p_0 + x_1p_1 + x_2p_2 = 0$ is tangent to the conic locus ${\mathcal C}(x,p)=0$
where
$$
{\mathcal C}(x,p)=
{\left|
  \begin{array}{cccc}
    0 & p_0 & p_1 & p_2 \\
    p_0 & \frac{\partial^2 P}{\partial x_0^2} & \frac{\partial^2 P}{\partial x_0 \partial x_1} & \frac{\partial^2 P}{\partial x_0 \partial x_2} \\
    p_1 & \frac{\partial^2 P}{\partial x_1 \partial x_0} & \frac{\partial^2 P}{\partial x_1 \partial x_1} & \frac{\partial^2 P}{\partial x_1\partial x_2} \\
    p_2 & \frac{\partial^2 P}{\partial x_2 \partial x_0} & \frac{\partial^2 P}{\partial x_2 \partial x_1} & \frac{\partial^2 P}{\partial x_2\partial x_2} \\
  \end{array}
\right|}
$$
\end{rem}

Set $\mathcal S_r=\mathcal T_r\cap\mathbb N^2.$ $\mathcal S_r=\mathcal S_r^1\cup\mathcal S_r^2,$ $\mathcal S_r^1=\{(x, y)\in\mathcal S_r: 0\leq x\leq r\}$ and
$\mathcal S_r^1=\{(x, y)\in\mathcal S_r: r< x\leq 2r\}.$ $dim\mathcal H_{3r}=card(\mathcal S_r^1)+card(\mathcal S_r^2).$
\begin{prop} \label{p1}
$$card(\mathcal S_r^1)=\left\{\begin{array}{ll}
\frac{3r^2+6r+4}{4}&\  \mbox{if}\ \  r \  \mbox{is  \ even}\\
&\\
\frac{3r^2+6r+3}{4}&\  \mbox{if}\ \  r \  \mbox{is  \ odd}.
\end{array} \right.$$
\end{prop}
\begin{proof}
Let $\alpha\in\{0,\cdots, r\}$ and set $\mathcal D_r^{\alpha}=\{(x,y)\in\mathcal S_r^1 : x=\alpha\}.$ Therefore
$$Card(\mathcal S_r^1)=\displaystyle{\sum_{\alpha=0}^r}Card(\mathcal D_r^{\alpha}).$$
Let $$\beta^{\alpha}_{max}=max\{\beta : (\alpha, \beta)\in\mathcal D_r^{\alpha}\}$$
and
$$\beta^{\alpha}_{min}=min\{\beta : (\alpha, \beta)\in\mathcal D_r^{\alpha}\}.$$
$Card(\mathcal D_r^{\alpha})=\beta^{\alpha}_{max}-\beta^{\alpha}_{min}+1.$\\
It is easy to prove that
$$Card(\mathcal D_r^{\alpha})=(\alpha+1)+\left\lfloor \frac{\alpha}{2} \right\rfloor
=\left\{\begin{array}{ll}
\frac{3\alpha+2}{2}&\  \mbox{if}\ \  \alpha \  \mbox{is  \ even}\\
&\\
\frac{3\alpha+1}{2}&\  \mbox{if}\ \  \alpha \  \mbox{is  \ odd}.
\end{array} \right.
.$$
The result follows from the summation of all $Card(\mathcal D_r^{\alpha}),$ $\alpha\in\{0,\cdots, r\}.$
\end{proof}
\begin{prop} \label{p2}
$$card(\mathcal S_r^2)=\left\{\begin{array}{ll}
\frac{3r^2}{4}&\  \mbox{if}\ \  r \  \mbox{is  \ even}\\
&\\
\frac{3r^2+1}{4}&\  \mbox{if}\ \  r \  \mbox{is  \ odd}.
\end{array} \right.$$
\end{prop}
\begin{proof}
Let $\alpha\in\{r+1,\cdots, 2r\}$ and set $\mathcal D_r^{\alpha}=\{(x,y)\in\mathcal S_r^2 : x=\alpha\}.$ Therefore
$$Card(\mathcal S_r^2)=\displaystyle{\sum_{\alpha=r+1}^{2r}}Card(\mathcal D_r^{\alpha}).$$
Let $$\beta^{\alpha}_{max}=max\{\beta : (\alpha, \beta)\in\mathcal D_r^{\alpha}\}$$
and
$$\beta^{\alpha}_{min}=min\{\beta : (\alpha, \beta)\in\mathcal D_r^{\alpha}\}.$$
$Card(\mathcal D_r^{\alpha})=\beta^{\alpha}_{max}-\beta^{\alpha}_{min}+1.$\\
It is easy to prove that
$$Card(\mathcal D_r^{\alpha})=(3r+1)-2\alpha+\left\lfloor \frac{\alpha}{2} \right\rfloor
=\left\{\begin{array}{ll}
\frac{6r+2-3\alpha}{2}&\  \mbox{if}\ \  \alpha \  \mbox{is  \ even}\\
&\\
\frac{6r+1-3\alpha}{2}&\  \mbox{if}\ \  \alpha \  \mbox{is  \ odd}.
\end{array} \right.
.$$
The result follows from the summation of all $Card(\mathcal D_r^{\alpha}),$ $\alpha\in\{r+1,\cdots, 2r\}.$
\end{proof}
\begin{theo} \
$$dim\mathcal H_{3(1+s)}=
%\left[(1+s)^2+(1+s)+1\right]+\frac{(1+s)(2+s)}{2}
 \frac{3}{2}s^2+\frac{9}{2}s+4$$
\end{theo}
\begin{proof}
This result is a direct consequence of propositions (\ref{p1}) and (\ref{p2}).
\end{proof}
\begin{cor}
 The Poincar\'e series of the algebras $\mathcal H$ is
The Poincar\'e series of the algebras $\mathcal H$ is
\begin{equation*}
P(\mathcal H, t)= \sum_{s\geq-1} \mathrm{dim}\left( \mathcal H_{3(1+s)}\right) t^{3(s+1)}=\frac{1+t^3+t^6}{(1-t^3)^3}.
\end{equation*}
\end{cor}

\section{$H$-invariant JPS in any dimension}
In order to formulate the problem in any dimension, let us remind some number-theoretic notions concerning the enumeration of nonnegative integer points in a polytope or more generally discrete volume of a polytope.
\subsection{Enumeration of integer solutions to linear inequalities}
In their papers \cite{bero, cosusa}, the authors study the problem of nonnegative integer solutions to linear inequalities as well as their relation with the enumeration of integer partitions and compositions.\\  Define the weight of a sequence $\lambda=(\lambda_0, \lambda_2,\cdots, \lambda_{n-1})$ of integers to be $|\lambda| = \lambda_0 +\cdots+\lambda_{n-1}.$ If sequence $\lambda$ of weight $N$ has all parts nonnegative, it is called a composition of $N$; if, in addition, $\lambda$ is a non increasing sequence, we call it a partition of $N$.\\
Given an $r\times n$ integer matrix $C = [c_{i,j} ],$ $(i, j)\in(\{-1\}\cup\mathbb Z/r\mathbb Z)\times\mathbb Z/n\mathbb Z,$ consider the set $S_C$ of nonnegative integer sequences
$\lambda=(\lambda_1, \lambda_2,\cdots, \lambda_n)$ satisfying the constraints
\begin{equation}\label{encond}
c_{i,-1}+c_{i,0}\lambda_0+c_{i,1}\lambda_1+\cdots+c_{i,n-1}\lambda_{n-1}\geqq 0, \ \ 0\leq r\leq n-1
\end{equation}
The associated full generating function is defined as follows:
\begin{equation}
F_C(x_0, x_2,\cdots, x_{n-1})=\displaystyle{\sum_{\lambda\in S_C}}x_{0}^{\lambda_0}x_{1}^{\lambda_1}\cdots x_{n-1}^{\lambda_{n-1}}.
\end{equation}
This function "encapsulates" the solution set $S_C$: the coefficient of $q^N$ in $F_C$($qx_0$, $qx_1$,$\cdots$, $qx_{n-1}$) is a "listing" (as the terms of a polynomial) of all nonnegative integer solutions to (\ref{encond}) of weight $N$ and the number of such solutions is the coefficient of $q^N$ in $F_C(q, q, \cdots, q).$
\subsection{Formulation of the problem in any dimension}
Let $\mathcal R=\mathbb C[x_0, x_1,\cdots, x_{n-1}]$ be the polynomial algebra with $n$ generators. For given $n-2$ polynomials $P_1, P_2,\cdots, P_{n-2}\in\mathcal R,$ one can associate the JPS $\pi(P_1,\cdots, P_{n-2})$ on $\mathcal R$ given by the formula:
$$\{f,g\}=\frac{df\wedge dg\wedge dP_1\wedge\cdots\wedge dP_{n-2}}{dx_0\wedge dx_1\wedge\cdots\wedge dx_{n-1}},$$
for $f, g\in\mathcal R.$\\
 We will denote by $P$ the particular Casimir $P=\displaystyle{\prod_{i=1}^{n-2}}P_i$ of the Poisson structure $\pi(P_1,\cdots, P_{n-2}).$ We suppose that each $P_i$ is homogeneous in the sense of $\tau$-degree.
\begin{prop}
Consider a JPS $\pi(P_1,\cdots, P_{n-2})$ given by homogeneous (in the sense of $\tau$-degree) polynomials $P_1,\cdots, P_{n-2}.$
If $\pi(P_1,\cdots, P_{n-2})$ is $H$-invariant, then
\begin{equation}
\tau-\varpi(\sigma\cdot P)=\tau-\varpi(P)=\left\{\begin{array}{ll}
\frac{n}{2}&\  \mbox{if}\ \  n \  \mbox{is  \ even}\\
&\\
0&\  \mbox{if}\ \  n \  \mbox{is  \ odd}.
\end{array} \right.
\end{equation}
where $P=P_1P_2\cdots P_{n-2}.$
\end{prop}
\begin{proof}
Let $i<j\in\mathbb Z/n\mathbb Z$ and consider the set $I_{i,j},$ formed by the integers $i_1<i_2<\cdots <i_{n-2}\in\mathbb Z/n\mathbb Z\setminus\{i, j\}.$ We denote by $S_{i, j}$ the set of all permutation of elements of $I_\{i, j\}.$ We have:
$$\begin{array}{lll}
\{x_i, x_j\}&=&(-1)^{i+j-1}\frac{dx_i\wedge dx_j\wedge dP_1\wedge dP_2\wedge\cdots\wedge dP_{n-2}}{dx_0\wedge dx_1\wedge dP_2\wedge\cdots\wedge dx_{n-1}}\\ \\
&=&(-1)^{i+j-1}\displaystyle{\sum_{\alpha\in S_{i, j}}}(-1)^{\mid\alpha\mid}\frac{\partial P_1}{\partial x_{\alpha(i_1)}}\dots\frac{\partial P_{n-2}}{\partial x_{\alpha(i_{n-2})}}.
\end{array}$$
From the $\tau$-degree condition\\ \mbox{} \\
 $i+j\equiv $($\tau$-$\varpi(P_1)-\alpha(i_1)$)$+\cdots+$($\tau$-$\varpi(P_{n-2})-\alpha(i_{n-2})$) $modulo \ n.$\\
We can deduce therefore that
$$\tau-\varpi(P_1\cdots P_{n-2})\equiv\frac{n(n-1)}{2} \ \ modulo \ n.$$
And we obtain the first part of the result. The second part is the direct consequence of facts that
$$\sigma\cdot\{x_i, x_j\}=\{x_{i+1}, x_{j+1}\}=(-1)^{i+j-1}\displaystyle{\sum_{\alpha\in S_{i, j}}}(-1)^{\mid\alpha\mid}\frac{\partial(\sigma\cdot P_1)}{\partial x_{{\alpha(i_1)}+1}}\dots\frac{\partial(\sigma\cdot P_{n-2})}{\partial x_{\alpha(i_{n-2})+1}},$$
$\alpha(i_1)+1\neq\cdots\neq\alpha(i_{n-2})+1\in Z/n\mathbb Z\setminus\{i+1, j+1\}$
and the $\tau$-degree condition.
\end{proof}
Set
$$l=\left\{\begin{array}{ll}
\frac{n}{2}&\  \mbox{if}\ \  n \  \mbox{is  \ even}\\
&\\
0&\  \mbox{if}\ \  n \  \mbox{is  \ odd}.
\end{array} \right.$$
Let $\mathcal H$ the set of all $Q=\in\mathcal R$
such that $\tau$-$\varpi(\sigma\cdot Q)=\tau$-$\varpi(Q)=l.$ One can easily check the following result:
\begin{prop}
$\mathcal H$ is a sub-vector space of $\mathcal R.$ It is subalgebra of $\mathcal R$ if $l=0.$
\end{prop}
We endow $\mathcal H$ with the usual grading of the polynomial algebra $\mathcal R.$
For $Q$ an element of $\mathcal R,$ we denote by $\varpi(Q)$ its usual weight degree.
We denote by $\mathcal H_i$ the homogeneous subspace of $\mathcal H$ of degree $i.$
\begin{prop}\label{condn}
If $n$ is not a divisor of $i$ (in other words $i\neq nm$ then $\mathcal H_i=0.$
\end{prop}
\begin{proof}
It is clear the $\mathcal H_0=\mathbb C.$ We suppose now $i\neq 0.$ Let $Q\in\mathcal H_i,$ $Q\neq 0.$ Then
$$Q=\displaystyle{\sum_{k_{1},\cdots, k_{i-1}}} a_{k_{1},\cdots, k_{i-1}}x_{k_1}\cdots x_{k_{i-1}}x_{l-k_{1}-\cdots-k_{i-1}}.$$
Hence
$$\sigma\cdot Q=\displaystyle{\sum_{k_{1},\cdots, k_{i-1}}} a_{k_{1},\cdots, k_{i-1}}x_{k_1+1}\cdots x_{k_{i-1}+1}x_{l-k_{1}-\cdots-k_{i-1}+1}.$$
Since $\tau$-$\varpi(\sigma\cdot Q)=\tau$-$\varpi(Q)=l,$ $i\equiv 0 \ modulo \ n.$
\end{proof}
Set $Q=\sum\beta x_0^{\alpha_0}x_1^{\alpha_1}\cdots x_{n-1}^{\alpha_{n-1}}.$  We suppose that $\varpi(Q)=n(1+s).$
We want to find all $\alpha_0, \alpha_1,\cdots, \alpha_{n-1}$ such that $Q\in\mathcal H$ and therefore the dimension $\mathcal H_{3(1+s)}$ as $\mathbb C$-vector space.
\begin{prop}
There exist $s_0, s_1\cdots, s_{n-1}$ such that
\begin{equation}\label{hjpsn}
\left\{\begin{array}{lcl}
\alpha_0+\alpha_1+\cdots+\alpha_{n-1}&=&n(1+s)\\
0\alpha_0+\alpha_1+2\alpha_2+\cdots(n-1)\alpha_{n-1}&=&l+ns_0\\
1\alpha_0+2\alpha_1+3\alpha_2+\cdots (n-1)\alpha_{n-2}+0\alpha_{n-1}&=&l+ns_1\\
&.&\\
&.&\\
&.&\\
(n-2)\alpha_0+(n-1)\alpha_1+0\alpha_2+\cdots (n-3)\alpha_{n-4}+(n-3)\alpha_{n-1}&=&l+ns_{n-2}\\
(n-1)\alpha_0+0\alpha_1+1\alpha_2+\cdots (n-3)\alpha_{n-3}+(n-2)\alpha_{n-1}&=&l+ns_{n-1}
\end{array}\right.
\end{equation}
\end{prop}
\begin{proof}
That is the direct consequence of the fact that $\tau$-$\varpi(\sigma\cdot Q)=\tau$-$\varpi(Q)=l.$
\end{proof}
One can easily obtain the following result:
\begin{prop}
The system equation (\ref{hjpsn}) has as solution:
\begin{equation}
\alpha_i=s_{n-i-1}-s_{n-i}+r, \ \ \ i\in\mathbb Z/n\mathbb Z
\end{equation}
where $r=s+1$ and the $s_0,\cdots s_{n-1}$ satisfy the condition
\begin{equation}
s_0+s_1+\cdots s_{n-1}=\frac{(n-1)n}{2}r-l.
\end{equation}
\end{prop}

Therefore $\alpha_0, \alpha_1,\cdots \alpha_{n-1}$ are completely determined by the set of nonnegative integer sequences $(s_0, s_1,\cdots, s_{n-1})$ satisfying the constraints
\begin{equation}\label{cons0}
c_i : s_{n-i-1}-s_{n-i}+r\geq 0, \ \ \ i\in\mathbb Z/n\mathbb Z
\end{equation}
and such that
\begin{equation}\label{cons}
s_0+s_1+\cdots+s_{n-1}=\frac{(n-1)n}{2}r-l.
\end{equation}
There are two approaches to determine the dimension of $\mathcal H_{nr}.$\\
The first one is exactly as in the case of dimension 3. The constraint (\ref{cons}) is equivalent to say that:
$$s_{n-1}=-(s_0+s_1+\cdots+s_{n-2})+\frac{(n-1)n}{2}r-l.$$
Therefore, by replacing $s_{n-1}$ by this value, $\alpha_0,\cdots, \alpha_{n-1}$ are completely determined by the set of nonnegative integer sequences $(s_0, s_1,\cdots, s_{n-2})$ satisfying the constraints:
\begin{equation}\label{cons1}
\left\{\begin{array}{ll}
c_{-1}' :&s_0+s_1+\cdots+s_{n-3}+s_{n-2}\leq \frac{(n-1)n}{2}r-l\\
c_0' : &2s_0+s_1+\cdots+s_{n-3}+s_{n-2}\leq \left[\frac{(n-1)n}{2}+1\right]r-l\\
c_1' :&s_0+s_1+\cdots+s_{n-3}+2+s_{n-2}\geq \left[\frac{(n-1)n}{2}k+1\right]r-l\\
c'_i :& s_{n-i-1}-s_{n-i}+r\geq 0, \ \ \ i\in\mathbb Z/n\mathbb Z\setminus\{0, 1\}.
\end{array}\right.
\end{equation}
Hence, the dimension $\mathcal H_{nr}$ is just the number of nonnegative integer points contain in the polytope given by the system (\ref{cons1}), where $r=s+1.$\\
In dimension 3, ones obtain the triangle in $\mathbb R^2$ given by the vertices $A(0, 2r),$ $B(r, 2r),$ $C(2r, 0)$ (see section 2).\\ \mbox{} \\
In dimension 4, we get the following polytope: \newpage
\begin{figure}[h!]
\centering
\includegraphics[width=8cm, height=9cm]{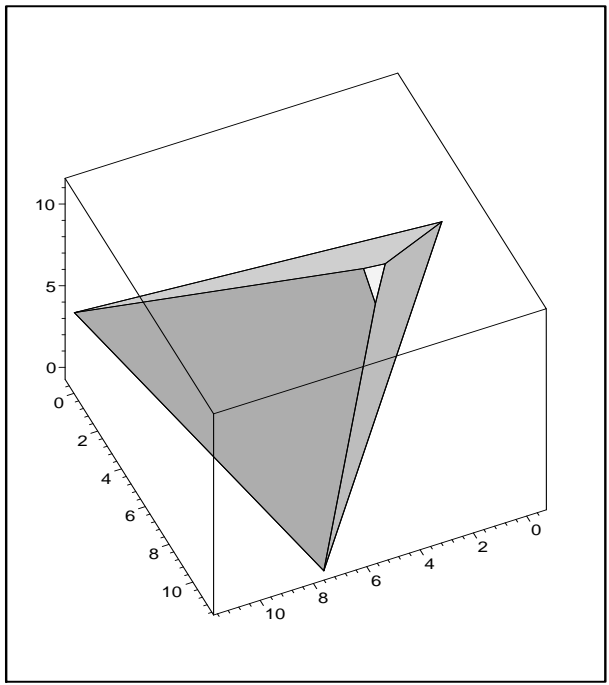}
\caption{An example of the polytope ${T}_4$ in the case $r=4$. The vertices are in
$
(-\frac{1}{2},r-\frac{1}{2},2r-\frac{1}{2}),
(r-\frac{2}{3},2r-\frac{2}{3},3r-\frac{2}{3}),
(r,2r-1,3r-1),
(r,2r,3r-2),
(2r-\frac{1}{2},3r-\frac{1}{2},-\frac{1}{2}),
(3r-\frac{1}{2},-\frac{1}{2},r-\frac{1}{2}).
$}\label{fig:Penta3D}
\end{figure}
For the second method, one can observe that the dimension of $\mathcal H_{nr)}$ is nothing else that the cardinality of the set $S_C$ of all compositions $(s_0,\cdots, s_{n-1})$ of $N=\frac{(n-1)n}{2}r-l$ subjected to the constraints (\ref{cons0}). Therefore if $S_C$ is the set of all nonnegative integers $(s_0,\cdots, s_{n-1})$ satisfying the constraints (\ref{cons0}) and $F_{C}$ is the associated generating function, then the dimension of $\mathcal H_{nr}$ is the coefficient of  $q^N$ in $F_C(q, q, \cdots, q).$ The set $S_C$ consists of all nonnegative integers points contained in the polytope of $\mathbb R^n$:
\begin{equation}
\begin{array}{ll}
\mathcal P_n :& \begin{array}{ll}
 x_{n-i-1}-x_{n-i}+r\geq 0, & i\in\mathbb Z/n\mathbb Z\\
  x_i\geq 0, & i\in\mathbb Z/n\mathbb Z
\end{array}
\end{array}
\end{equation} \newpage
\begin{figure}[h!]
\centering
\includegraphics[width=13cm, height=14cm]{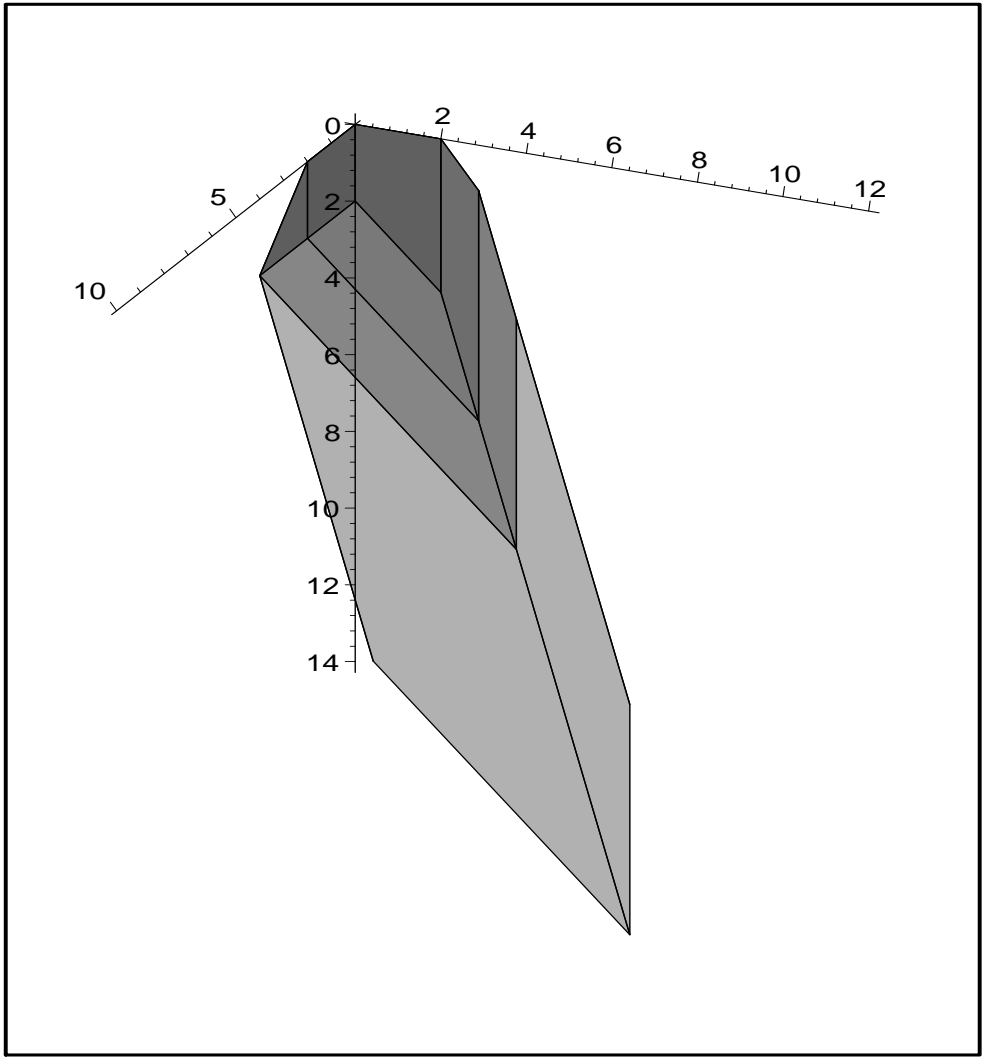}
\caption{An example of the polytope ${\mathcal P}_3$ in the case $r=2$.
}\label{fig:INF3D}
\end{figure}
\subsubsection*{Acknowledgements}
Authors are grateful to M. Beck and T.Schedler for useful and illuminative discussions. This  work has begun when G.O. was in visiting Mathematics Research Unit at Luxembourg. G.O. thanks this institute for the invitation and for the kind hospitality. A Part of this work has been done when S.P: was in visiting Max Planck Institute at Bonn. S.P. thanks this institute for the invitation and for good working conditions. S.P. have been partially financed by ``Fond National de Recherche (Luxembourg)''. He is thankful to LAREMA for a kind invitation and a support during his stay in Angers. V.R. was partially supported by the French National Research Agency (ANR) Grant 2011 DIADEMS and by franco-ukrainian PICS (CNRS-NAS) in Mathematical Physics. He is grateful to MATPYL project for a support of T.Schedler visit in Angers and to the University of Luxembourg for a support of his visit to Luxembourg.

\end{document}